\documentclass[onecolumn,12pt]{IEEEtran}
\usepackage{algorithmic}
\usepackage{amsmath}
\usepackage{slashbox}
\usepackage{latexsym}
\usepackage{epsfig}
\usepackage{amsmath}
\usepackage{multicol}
\usepackage{amssymb}
\usepackage{subfigure}
\usepackage{cite}
\usepackage{setspace}
\def\E{{\rm E}}
\def\C{{\rm C}}

\newtheorem{theorem}{Theorem}[section]

\newtheorem{proposition}[theorem]{Proposition}

\newenvironment{proof}[1][Proof]{\begin{trivlist}
\item[\hskip \labelsep {\bfseries #1}]}{\end{trivlist}}

\newcommand{\qed}{\nobreak \ifvmode \relax \else
      \ifdim\lastskip<1.5em \hskip-\lastskip
      \hskip1.5em plus0em minus0.5em \fi \nobreak
      \vrule height0.75em width0.5em depth0.25em\fi}
\title{Quantization for Uplink Transmissions in Two-tier Networks with Femtocells}
\author{Chan Dai Truyen Thai and Marion Berbineau}
\begin{document}
\maketitle

\begin{abstract}
We propose two novel schemes to level up the sum--rate for a two-tier network with femtocell where the backhaul uplink and downlink connecting the Base Stations have limited capacity. The backhaul links are exploited to transport the information in order to improve the decoding of the macrocell and femtocell messages. In the first scheme, Quantize-and-Forward, the Femto Base Station (FBS) quantizes what it receives and forwards it to the Macro Base Station (MBS). Two quantization methods are considered: Elementary Quantization and Wyner-Ziv Quantization. In the second scheme, called Decode-and-Forward with Quantized Side Information (DFQSI) to be distinguished with the considered conventional Decode-and-Forward (DF) scheme. The DFQSI scheme exploits the backhaul downlink to quantize and send the information about the message in the macrocell to the FBS to help it better decode the message, cancel it and decode the message in the femtocell. The results show that there are interesting scenarios in which the proposed techniques offer considerable gains in terms of maximal sum rate and max minimal rate.
\end{abstract}
\section{Introduction}
\subsection{Motivation}
Femtocells are used to extend the service coverage to indoor areas and improve the capacity of cellular systems. In a two-tier network with femtocells, a femtocell Base Station (FBS) is connected to the Base Station of the macrocell (MBS) via a high-rate link. A two-tier network consists of interconnected macrocells and femtocells. If a macrocell and its femtocells use different frequency bands, then indoor areas can be well covered, but the overall spectral efficiency may deteriorate due to the need of additional bandwidth. In case they use the same frequency band, system capacity can be improved if the system is carefully designed to minimize the interference among different cells \cite{hist}.

There are several types of interference-related issues in a femtocell network. A cross-tier interference occurs between a macrocell and a femtocell, while intra-tier interference arises among femtocells when femtocells are installed close to each other. A known issue in systems with femtocells is the near-far problem, explained as follows. A femtocell is often installed around the edge of the macrocell since the signal from the macro BS is attenuated in this area, especially for indoor terminals. In this case, a nearby macrocell user needs to raise its transmit power in an uplink transmission in order to reach its remote MBS. Hence, it creates interference to the femtocell. For macrocell downlink transmission, the MBS also has to raise its transmit power to reach the users at the cell edge, which also increases its interference to the receivers in the femtocell.

\subsection{Related work and Contribution}
Femtocell networks have received a considerable research attention. \cite{asurvey} presents a thorough survey on femtocells including technical challenges and commercial aspects.~\cite{requirements} focuses on femtocell requirements and network architecture. Other discussed femtocell standardization history~\cite{standardization}\cite{standardization2}, the necessity of femtocells and its advantages and disadvantages~\cite{industry}.

The existing works present several methods to cope with the interference in a femtocell network, such as: radio resource management, power control, access management, femtocell location management, antenna and other solutions. As in a radio resource management solution, carriers (in CDMA~\cite{autonomous}) or sub-carriers (in OFDMA~\cite{ofdma}\cite{avoidance}) are manually or dynamically arranged to optimally reduce the interference~\cite{spectrum}. However, power control in femtocell or macrocell is often combined with frequency channel management to reduce interference~\cite{optimization}\cite{hspa}\cite{downlink}. In one of the hybrid solutions, interference-limited coverage area (ILCA) of a femtocell is defined as an area around the FBS in which interference to a femtocell is lower than an adjustable threshold. ILCA depends on the position of the femtocell within the macrocell. If the ILCA radius of a femtocell is larger than a certain value the same channel with the macrocell is used. Otherwise another channel is used~\cite{hybrid}. In~\cite{uplink} different methods for interference mitigation for uplinks can be combined, such as receive antenna sectorization, two-tier network spectrum split, femtocell exclusion region, and a tier selection based femtocell handoff. It has been shown that completely closed or completely open access policies are not optimal~\cite{dealing}. Antenna related methods are also proposed in which multi-antenna~\cite{multiantenna} or switching among antennas~\cite{switching} is considered.

Other interference mitigation methods have also been explored, such as user-assisted coverage and interference optimization~\cite{user-assisted}, mobility management~\cite{mobility}, minimizing pilot leakage outside the femtocell building~\cite{leakage}, considering statistics of time periods for a signal travel from a macro user to its base station and a FBS~\cite{statistics}, pilot sensing to select frequency bands for femtocells among frequency bands of macro networks~\cite{pilotsensing}.

A scheme coping with the unreliability of the backhaul link is proposed in \cite{robust}. In the scheme, the message from each of the femto users is divided into layers. The number of message layers delivered to the MBS depends on the current status of the backhaul link. The MBS uses Multiple Access Channel and Successive Interference Cancellation to decode the desired messages. However, the paper considers real rather than complex channel gains. In a previous publication \cite{chanthai}, we also propose to use the backhaul link to send the information of interference to the FBS so that it can better decode the received signal.

If at a certain time, there are an uplink and a downlink for two users using the same relay node. The transmissions can be combined in an overhearing scheme \cite{overhearing_fansun_comlett,dof_fansun_spawc, mmse_fansun_siglett1}. On the other hand, if there are two-hop users using different relay nodes: one has an uplink and one has a downlink, we have a multi-way scheme \cite{fourway_huaping_spawc, fourway_huaping_lett, mmse_fansun_siglett2}. If we use the CDR, overhearing and multi-way schemes as described above to respective cases instead of the corresponding conventional schemes, several time slots are saved and a significant improvement is gained.

Regarding Quantize-and-Forward relaying, \cite{qf1,qf2} consider in-band relaying where the destination receives the signals from the sources and from the relay in the same band whereas the relaying of the FBS can be considered as out-of-band relaying \cite{outband}. In this paper we propose novel quantization schemes for out-of-band relaying. Specifically, we propose two schemes to exploit the backhaul links between the MBS and the FBS and facilitate the decoding of macro and femto uplink messages. In the first scheme, Quantize and Forward (QF), the FBS quantizes the received signal and forwards it to the MBS. Two quantization methods are considered: Elementary Quantization (EQ) and Wyner-Ziv Quantization (WZQ). The difference between EQ and WZQ is that in WZQ, the FBS takes into account the signal that the MBS receives from the macro user. The second proposed scheme is a Decode-and-Forward (DF) scheme based on a conventional DF scheme. However, in the proposed scheme, we exploit the backhaul downlink to facilitate the decoding of the uplink message in the femtocell and thus enable to increase the achievable rates. In the following sections, we present the conventional DF scheme, denoted as DF, as well as the proposed DF scheme, denoted as DFQSI. All schemes will be compared in terms of the maximum sum--rate and the minimum rate of two users.

The rest of the paper is organized as follows. Section \ref{system_model} introduces the system model. Section \ref{uplink_only} presents the DF and QF schemes for the case with backhaul uplink only while the QF scheme in the case with backhaul uplink and downlink is presented in Section \ref{uplink_and_downlink}. Section \ref{performance_metric} introduces the performance metric used. Section \ref{numerical_results} presents and analyses the numerical results. Section \ref{conclusion} concludes the paper.
\section{System Model}\label{system_model}
\begin{figure}
\centering
\includegraphics[width=0.5\columnwidth]{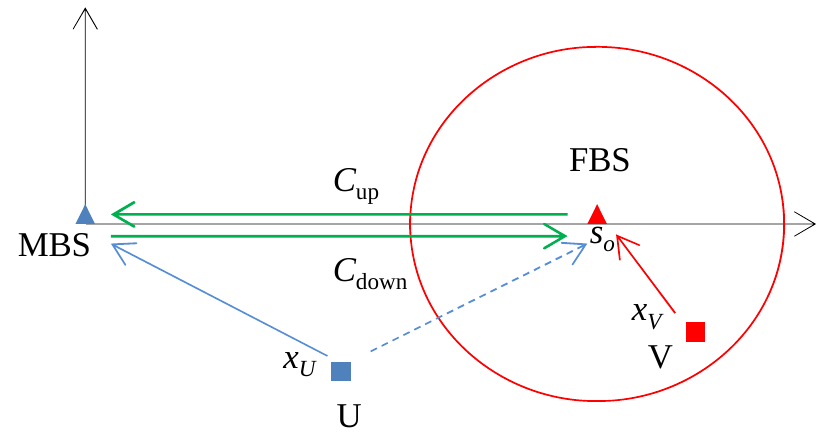}
\caption{We consider a two-tier network with one macrocell and one femtocell each with one user.}
\label{network1}
\end{figure}
Consider a two-tier network with one macrocell and one femtocell as in Fig. \ref{network1}. The MBS and the FBS are located at $(0,0)$ and $(s_o,0)$ of the coordinate system. The macrocell user U and the femtocell user V are randomly located inside a circle with the corresponding BS as the center and radius $r_M$ and $r_F$ respectively. The uplink messages of the users U and V are denoted by $x_U$ and $x_V$, and they are destined to the MBS B and FBS F, respectively. The MBS and FBS are connected through a two-way backhaul link with capacity $C_\mathrm{up}$  and $C_\mathrm{down}$ (b/s/Hz) which are known at all stations. We assume that the uplink message from the femto user V to the FBS will be forwarded to the final destination via the MBS.

The reciprocal channel between station $i$ and $j$ is denoted by $h_{ij}$ and known at all stations, $~i\in\{U, V\}, ~j\in\{F, B\}$. The received signal and Additive White Gaussian Noise (AWGN) at station $j$ are denoted by $y_{j}$ and $z_{j}\sim\mathcal{CN}(0, \sigma^2)$ respectively. We therefore consider that a channel coefficient is given by $h=\sqrt{\frac{K}{d^{\alpha}}}g$, $K=\left(\frac{\lambda}{4\pi}\right)d_{o}^{\alpha-2}$, where $\alpha$ is the path loss exponent; $K$ is a constant depending on the used frequency and the reference distance $d_{o}$; $d$ is the distance from the source to the destination; $10\log_{10}g$ is the log-normal shadowing coefficient \cite{uplink}. The distance from the femto user to the MBS is long. In addition, the femto user is very near the FBS, its transmit power is adjusted to a very low level. From all this, it is viable to assume that $h_{VB}=0$ for the rest of the paper.

All wireless transmissions are in the same frequency band. Every station has a single antenna. We assume perfect power control for the transmissions from the macro user U to its BS and from the femto user V to its BS. By this assumption, the transmit power is equal to the inverse of the channel magnitude such that the received signal at the corresponding receiver has a power of $P_R$ \cite{uplink}.

We use $\triangleq$ to define a new variable or function. Denote $\gamma_{ij} \triangleq \frac{P_i|h_{ij}|^2}{\sigma^2}$, where $P_i$ is the transmit power of station $i$, $i\in\{U, V\}, ~j\in\{F, B\}$, and $\C(\gamma_{ij}) \triangleq \log_2(1 + \gamma_{ij})$ as the maximal achievable rate for a transmission over the channel $h_{ij}$ without interference. The signals are denoted by using small letters $(x,y,z)$ and the correspondent random variables are denoted using capitalized letters $(X,Y,Z)$. When a station quantizes signal $x$ and sends the quantized version to another station over a backhaul link, the received signal at the receiver is written $\hat{x} = x + z_Q$ in which the quantization noise $z_Q$ is assumed to be Gaussian \cite{secrecy}. The variance of the quantization noise depends on the capacity of the backhaul link.
\section{Relaying with the Backhaul Uplink Only}\label{uplink_only}
This section presents two relaying modes, DF and QF, when only the backhaul uplink is used ($C_\mathrm{up} > 0$, $C_\mathrm{down} = 0$). The DF scheme in this section is the conventional DF. The DFQSI scheme refers to the DF scheme when both backhaul uplink and downlink are used and will be presented in the next section.

When U sends message $x_U$ to the MBS and V sends message $x_V$ to the FBS. The FBS and MBS respectively receive the signals
\begin{equation}\label{yfb}\left\{
\begin{array}{l}
	y_F = h_{UF}x_U + h_{VF}x_V + z_F,\\
	y_B = h_{UB}x_U + z_B.
\end{array}\right.
\end{equation}
where $z_F$ and $z_B$ are AWGN at the FBS and MBS respectively. The achievable rates $R_U, R_V$ are the uplink rates from U and V respectively that the MBS can decode with help from the FBS over the backhaul link.

The transmission rates of messages $x_U$ and $x_V$ can be calculated and selected in several ways. First, the MBS or the FBS can decode $x_U$ treating the contribution of $x_V$ as noise, cancels the contribution of $x_U$ and decodes $x_V$. This decoding order is denoted as UV. The opposite order, VU, is decoding $x_V$ treating the contribution of $x_U$ as noise first. These two decoding orders give different transmission rate pairs for $x_U$ and $x_V$. Knowing all channels, all stations are able to calculate the rates of $x_U$, $x_V$ before the transmissions start. Which order is selected depends on the requirement of the rates or the priority of the users. In DF, decoding happens at the FBS. Hence, UV or VU refers to the two decoding order at the FBS. In the DF scheme, the FBS decodes $x_V$ and forwards it to the MBS over the backhaul uplink. In the QF scheme, the FBS quantizes the received signal and forwards it to the MBS.
\subsection{Decode and Forward}\label{daf1}
In DF, the message $x_U$ is decoded at FBS and forwarded to the MBS through the backhaul link. Hence the transmission rate is limited to $C_\mathrm{up}$. We consider the following two cases.
\begin{itemize}
\item The FBS decodes $x_U$ from $y_F$ treating $x_V$ as noise. The achievable rate of $x_U$ therefore satisfies
\begin{equation}
R_U \leq \log_2\left(1 + \frac{|h_{UF}|^2}{|h_{VF}|^2 + \sigma^2}\right)\ = \C\left(\frac{\gamma_{UF}}{\gamma_{VF} + 1}\right).
\end{equation}
Thus $x_U$ is cancelled and $x_V$ is decoded with rate $R_V \leq \C(\gamma_{VF})$. On the other hand, at the MBS, $x_U$ is decoded with constraint $R_U \leq \C(\gamma_{UB})$. With no quantization, the message $x_V$ sent over the backhaul uplink with capacity $C_\mathrm{up}$ has to satisfy $R_V \leq C_\mathrm{up}$. Finally,
\begin{equation}
R^\mathrm{DF-UV}_V = \min\left\{C_\mathrm{up},\C\left(\gamma_{VF}\right)\right\}.
\end{equation}
\begin{equation}
R^\mathrm{DF-UV}_U = \min\left\{\C\left(\frac{\gamma_{UF}}{\gamma_{VF} + 1}\right), \C\left(\gamma_{UB}\right)\right\},
\end{equation}
\item The FBS decodes $x_V$ from $y_F$ treating $x_U$ as noise. The achievable rate of $x_V$ therefore satisfies
\begin{equation}
R_V \leq \log_2\left(1 + \frac{|h_{VF}|^2}{|h_{UF}|^2 + \sigma^2}\right)\ = \C\left(\frac{\gamma_{VF}}{\gamma_{UF} + 1}\right)
\end{equation}
The uplink rate of the femto user is selected as
\begin{equation}\label{rv_df1_vu}
R^\mathrm{DF-VU}_V = \min\left\{C_\mathrm{up},\C\left(\frac{\gamma_{VF}}{\gamma_{UF} + 1}\right)\right\}.
\end{equation}
The uplink rate of the macro user is therefore selected as
\begin{equation}\label{ru_df1_vu}
R^\mathrm{DF-VU}_U = \C\left(\frac{|h_{UB}|^2}{\sigma^2}\right) = \C\left(\gamma_{UB}\right).
\end{equation}
\end{itemize}

\subsection{Quantize and Forward}\label{qaf}
In this part, the FBS quantizes its received signal $y_F$ to $\hat{y}_F$ and forwards it to the MBS. In EQ, the signal is quantized by not taking into account that the MBS has a side information through the observation of $y_B$. Therefore, $\hat{y}_F$ contains a redundant information. Another option is the WZQ which compresses the signals by taking into account that $y_B$ is already available at the MBS. Eventually this more efficient compression will result in higher achievable rates for U and V \cite{secrecy}.

The MBS receives from the FBS
\begin{equation}\label{yhatf} \hat{y}_F = y_F + z_{Q_Y} = h_{UF}x_U + h_{VF}x_V + z_F + z_{Q_Y} \end{equation}
where $z_{Q_Y}$ is the quantization noise. We consider two possible decoding orders.

\subsubsection{QF-UV}

\begin{proposition}\label{prop_qf_uv}
If the decoding order at the MBS is UV, the achievable rates for macro and femto uplinks are
\begin{equation}\label{ruv_qf_uv}\left\{
\begin{array}{l}
	R^\mathrm{QF-UV}_U = \C\left(\gamma_{UB} + \frac{\gamma_{UF}}{\gamma_{VF} + 1 + \beta}\right)\\
	R^\mathrm{QF-UV}_V = \C\left(\frac{\gamma_{VF}}{1 + \beta}\right)
\end{array}\right.
\end{equation}
where $\beta = \frac{\sigma^2_Q}{\sigma^2} = \frac{\gamma_{UF} + \gamma_{VF} + 1}{2^{C_\mathrm{up}} - 1}$ in case of EQ and $\beta = \frac{\gamma_{VF} + 1}{2^{C_\mathrm{up}} - 1} + \frac{\gamma_{UF}}{(2^{C_\mathrm{up}} - 1)(\gamma_{UB} + 1)}$ in case of WZQ.
\end{proposition}
\begin{proof}
From $y_B$ and $\hat{y}_F$, the MBS first decodes $x_U$. $R_U$ therefore satisfies
\begin{equation}
R_U \leq \log_2\left(1 + \frac{|h_{UB}|^2}{\sigma^2} + \frac{|h_{UF}|^2}{|h_{VF}|^2 + \sigma^2 + \sigma^2_{Q_Y}}\right) = \log_2\left(1 + \gamma_{UB} + \frac{\gamma_{UF}}{\gamma_{VF} +1 + \beta}\right).
\end{equation}
with $\beta = \frac{\sigma^2_{Q_Y}}{\sigma^2}$. The MBS cancels the contribution of $x_U$ in $\hat{y}_F$ and decodes $x_V$ with
\begin{equation}
R_V \leq \log_2\left(1 + \frac{|h_{VF}|^2}{\sigma^2 + \sigma^2_{Q_Y}}\right) = \log_2\left(1 + \frac{\gamma_{VF}}{1 + \beta}\right).
\end{equation}
Now we calculate $\beta$ using two methods of quantizations.
\begin{itemize}
\item For EQ, the backhaul link has to carry an amount of information of $I(Y_F; \hat{Y}_F) =  \C{\left(\frac{\E[|y_F|^2]}{\sigma^2_{Q_Y}}\right)}$. Setting this to $C_\mathrm{up}$, we have the quantization noise as
\begin{equation}
\sigma^2_{Q_Y} =  \frac{\E[|y_F|^2]}{2^{C_\mathrm{up}} - 1} = \frac{|h_{UF}|^2 + |h_{VF}|^2 + \sigma^2}{2^{C_\mathrm{up}} - 1} = \sigma^2 \frac{\gamma_{UF} + \gamma_{VF} + 1}{2^{C_\mathrm{up}} - 1}
\end{equation}
Thus
\begin{equation}
\beta =  \frac{\sigma^2_{Q_Y}}{\sigma^2} = \frac{\gamma_{UF} + \gamma_{VF} + 1}{2^{C_\mathrm{up}} - 1}
\end{equation}

\item For WZQ, the FBS quantizes $y_F$ taking into account that the MBS knows $y_B$, therefore the capacity of the backhaul link $C_\mathrm{up}$ is allocated to
\begin{equation}\label{I_Q}
\begin{array}{lll}
		I(Y_F; \hat{Y}_F|Y_B) &=& I(\hat{Y}_F; Y_F, Y_B) - I(\hat{Y}_F;Y_B)\\
		
		&=& I(\hat{Y}_F;Y_F) + I(\hat{Y}_F;Y_B|Y_F) - I(\hat{Y}_F;Y_B) \\
		
		&=& I(\hat{Y}_F;Y_F) - I(\hat{Y}_F;Y_B,X_U) + I(\hat{Y}_F;X_U|Y_B) \\
		
		&=& I(\hat{Y}_F;Y_F) - I(\hat{Y}_F;Y_B,X_U) + I(X_U;\hat{Y}_F,Y_B) - I(X_U;Y_B).
\end{array}
\end{equation}
The third equality is due to $I(Y_B; \hat{Y}_F|Y_F)$ = 0 since $Y_B - Y_F - \hat{Y}_F$ is a Markov chain. We have
\begin{equation}
I(\hat{Y}_F;Y_F) = \C\left(\frac{|h_{UF}|^2 + |h_{VF}|^2 + \sigma^2}{\sigma^2_{Q_Y}}\right) = \log_2\left(1 + \frac{\gamma_{UF} + \gamma_{VF} + 1}\beta\right),
\end{equation}
\begin{equation}
I(\hat{Y}_F;Y_B,X_U) = I(\hat{Y}_F;X_U) = \C\left(\frac{\gamma_{UF}}{\gamma_{VF} + 1 + \beta}\right),
\end{equation}
\begin{equation}
I(X_U;\hat{Y}_F,Y_B) = \C\left(\gamma_{UB} + \frac{\gamma_{UF}}{\gamma_{VF} + 1 + \beta}\right),
\end{equation}
\begin{equation}
I(X_U;Y_B) = \C\left(\gamma_{UB}\right).
\end{equation}
Substituting the above to (\ref{I_Q}), we have
\begin{equation}
I(Y_F; \hat{Y}_F|Y_B) = \C\left(\frac{(\gamma_{VF} + 1)(\gamma_{UB} + 1) + \gamma_{UF}}{\beta(\gamma_{UB} + 1)}\right).
\end{equation}
We use the uplink backhaul with capacity $C_\mathrm{up}$ to transmit the information above therefore $I(Y_F; \hat{Y}_F|Y_B) = C_\mathrm{up}$. Solving the equation, we have
\begin{equation}
\beta = \frac{\gamma_{VF} + 1}{2^{C_\mathrm{up}} - 1} + \frac{\gamma_{UF}}{(2^{C_\mathrm{up}} - 1)(\gamma_{UB} + 1)}.
\end{equation}
\end{itemize}
\end{proof}
\subsubsection{QF-VU}
\begin{proposition}\label{prop_qf_vu}
If the decoding order at the MBS is $x_V$, $x_U$, the achievable rates for macro and femto uplinks are
\begin{equation}\label{ruv_qf_vu}\left\{
\begin{array}{l}
	R^\mathrm{QF-VU}_U = \C\left(\gamma_{UB} + \frac{\gamma_{UF}}{1 + \beta}\right)\\
	R^\mathrm{QF-VU}_V = \C\left(\frac{\gamma_{VF}}{\frac{\gamma_{UF}}{\gamma_{UB} + 1} + 1 + \beta}\right)
\end{array}\right.
\end{equation}
where $\beta = \frac{\sigma^2_{Q_Y}}{\sigma^2} = \frac{\gamma_{UF} + \gamma_{VF} + 1}{2^{C_\mathrm{up}} - 1}$ in case of EQ and $\beta = \frac{\gamma_{VF} + 1}{2^{C_\mathrm{up}} - 1} + \frac{\gamma_{UF}}{(2^{C_\mathrm{up}} - 1)(\gamma_{UB} + 1)}$ in case of WZQ.
\end{proposition}

\begin{proof}
From $y_B$ and $\hat{y}_F$, the MBS first decodes $x_V$ using MMSE-SIC as in \cite{icc2011} to get $R^\mathrm{QF-VU}_V$ in (\ref{ruv_qf_vu}). The quantization and other steps are conducted similarly to the ones of Proposition \ref{prop_qf_uv}.
\end{proof}
\section{Relaying with Backhaul Uplink and Downlink}\label{uplink_and_downlink}
In this part we introduce a different transmission scheme based on the following observation. By exploiting the backhaul downlink, the MBS can send some information about $x_U$ or $y_B$ to the FBS in order to help FBS to decode $x_V$ when DF is used or quantize $y_F$ in case of QF. However, when the FBS quantizes $y_F$ using WZQ, it does not need to have $y_B$. With Gaussian signals, even knowning $y_B$ competely at the transmitter side cannot contribute to a better quantization of $y_F$ \cite{kramer}. We therefore do not consider another version of QF in this section but only DF.

The MBS first decodes $x_U$ from $y_B$. Then it quantizes and sends $\hat{x}_U$ to the FBS; then FBS uses this side information in order to decode $x_U$, which enables to achieve higher rates:
\begin{equation}
\hat{x}_U = x_U + z_{Q_U}.
\end{equation}
Denote $\gamma_{Q_U} = \frac{|x_U|^2}{\sigma^2_{Q_U}}$ as the SNR of the received signal over the downlink backhaul at the FBS. The downlink backhaul with capacity $C_\mathrm{down}$ is used to send an amount of information which is equal to
\begin{equation}
I(x_U, \hat{x}_U) = \frac{|x_U|^2}{\sigma^2_{Q_U}}.
\end{equation}
Thus
\begin{equation}
\gamma_{Q_U} = 2^{C_\mathrm{down}} - 1.
\end{equation}
Combining $y_F$, in which $x_V$ is treated as noise, and $\hat{x}_U$, the FBS decodes $x_U$ with a rate
\begin{equation}
R_U \leq \C\left(\gamma_{Q_U}+ \frac{\gamma_{UF}}{\gamma_{VF} + 1 + \beta}\right)
\end{equation}
in which $\gamma_{Q_U}$ depends on how the quantization is conducted. Cancelling the contribution of $x_U$ in $y_F$, the FBS decodes and forwards $x_V$ to the MBS. Similar to Section \ref{daf1}, we have
\begin{equation}
R^\mathrm{DFQSI-UV}_V = \min\left\{C_\mathrm{up},\C\left(\gamma_{VF}\right)\right\}.
\end{equation}
\begin{equation}
R^\mathrm{DFQSI-UV}_U = \min\left\{\C\left(\gamma_{Q_U} + \frac{\gamma_{UF}}{\gamma_{VF} + 1}\right), \C\left(\gamma_{UB}\right)\right\}.
\end{equation}
In the second case, when the FBS decodes $x_V$ treating $x_U$ as noise, the scheme is exactly the same as the scheme in section \ref{daf1} therefore $R^\mathrm{DFQSI-VU}_U = R^\mathrm{DF-VU}_U$ and $R^\mathrm{DFQSI-VU}_V = R^\mathrm{DF-VU}_V$.
\section{Performance Metric}
\label{performance_metric}
In any scheme, two decoding orders lead to two rate points $(R_U, R_V)$. The two points can make a pentagon or rectangular rate region as shown in Fig. \ref{rate_region}. Any point on the line connecting two points can be achieved applying a suitable time sharing. Therefore, the \textit{maximum sum-rate} corresponds to the point with a higher sum-rate.

The \textit{maximum of the minimum rate of two users} corresponds to the intersection of the line $R_U = R_V$ and the outer line of the rate region. This can be proved by reasoning as follows. Assume that the intersection A ($R^A, R^A$) of the line $R_U = R_V$ and the outer line of the rate region is not the point B ($R^B_{U}, R^B_{V}$) of the maximum minimum rate. This means that when moving along the outer line of the rate region from A we can see B such that $R^B_{U} > R^A$ and $R^B_{V}> R^A$. It is equivalent to the fact that the rate region is not a convex polygon. However, a rate region is always a convex hull of all achievable rate points and therefore a polygon. Consequently, the assumption is wrong and thus A and B coincide.
\begin{figure}
\centering
\includegraphics[width=1\columnwidth]{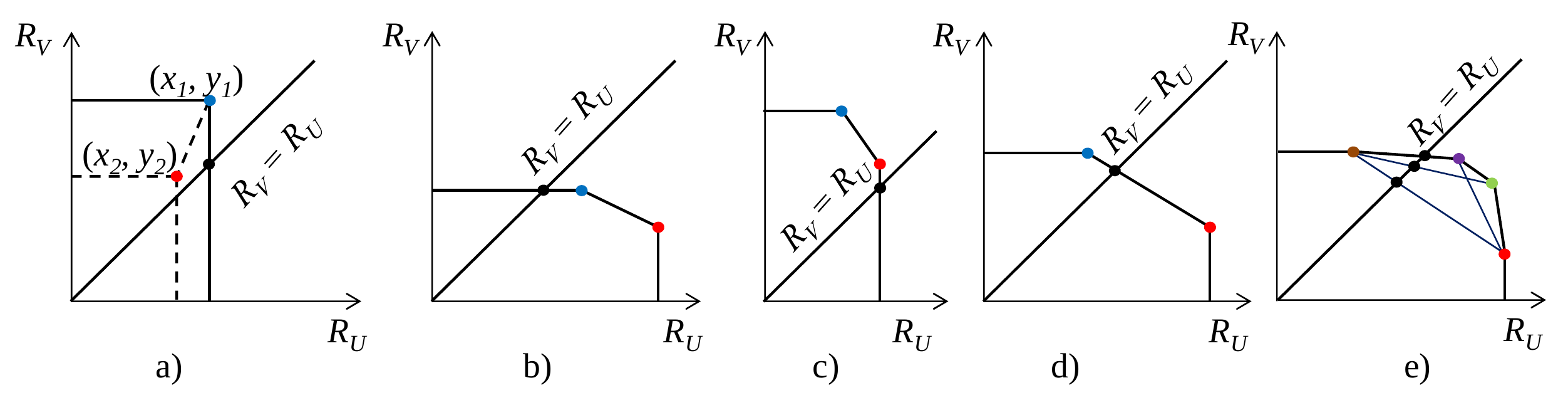}
\caption{To determine the maximum minimum rate of two users, several cases of rate regions are considered.}
\label{rate_region}
\end{figure}
\begin{proposition}\label{max_min_rate}
If a rate region of a scheme is determined by two points ($x_1$, $y_1$) and ($x_2$, $y_2$), the maximum of the minimum of two rates is given by
\begin{equation}\label{max_min_rate_eq}
\begin{array}{l}
\max_{(R_U,R_V)} \min(R_U,R_V)\\ = \mathrm{f}(x_1, x_2, y_1, y_2) \triangleq \left\{
\begin{array}{lll}
\min(\max(x_1, x_2), \max(y_1, y_2)) &\mbox{if}& (x_1 - x_2)(y_1 - y_2) > 0\\
\max(\min(x_1, y_1), \min(x_2, y_2)) &\mbox{if}& \left\{
\begin{array}{l}
(x_1 - x_2)(y_1 - y_2) < 0\\
(x_1 - y_1)(x_2 - y_2) > 0
\end{array}\right.\\
\frac{x_1y_2 - y_1x_2}{y_2 - y_1 + x_1 - x_2} &\mbox{if}& \left\{
\begin{array}{l}
(x_1 - x_2)(y_1 - y_2) < 0\\
(x_1 - y_1)(x_2 - y_2) < 0.
\end{array}\right.\\
\end{array}\right.
\end{array}
\end{equation}
\end{proposition}
\begin{proof}
The first condition in (\ref{max_min_rate_eq}) corresponds to the case when one point is ``inside'' another point as demonstrated in Fig. \ref{rate_region}a). The condition is ($x_1 < x_2$ and $y_1 < y_2$) or ($x_2 < x_1$ and $y_2 < y_1$) which is equivalent to $(x_1 - x_2)(y_1 - y_2) > 0$. The minimum rate is therefore $\min(\max(x_1, x_2), \max(y_1, y_2))$.

The second condition in (\ref{max_min_rate_eq}) corresponds to the case when two points are at one side of the line $R_U = R_V$ as demonstrated in Fig. \ref{rate_region}b) and c). The third condition in (\ref{max_min_rate_eq}) corresponds to the case when each point is at one side of the line $R_U = R_V$ as demonstrated in Fig. \ref{rate_region}d). The minimum rate of these cases are easily obtained.
\end{proof}

If a rate region is formed based on more than two rate points, we select the maximum of all minimum rates of each point pair which are formed by combining two points out of all the points. It is the highest point among the black points as in shown in Fig. \ref{rate_region}e)
\begin{equation}
\max_{(R_U,R_V)} \min(R_U,R_V) = \max_{i,j}\mathrm{f}(x_i, x_j, y_i, y_j)
\end{equation}
\section{Numerical Results}
\label{numerical_results}
\begin{figure}
\centering
\includegraphics[width=0.6\columnwidth]{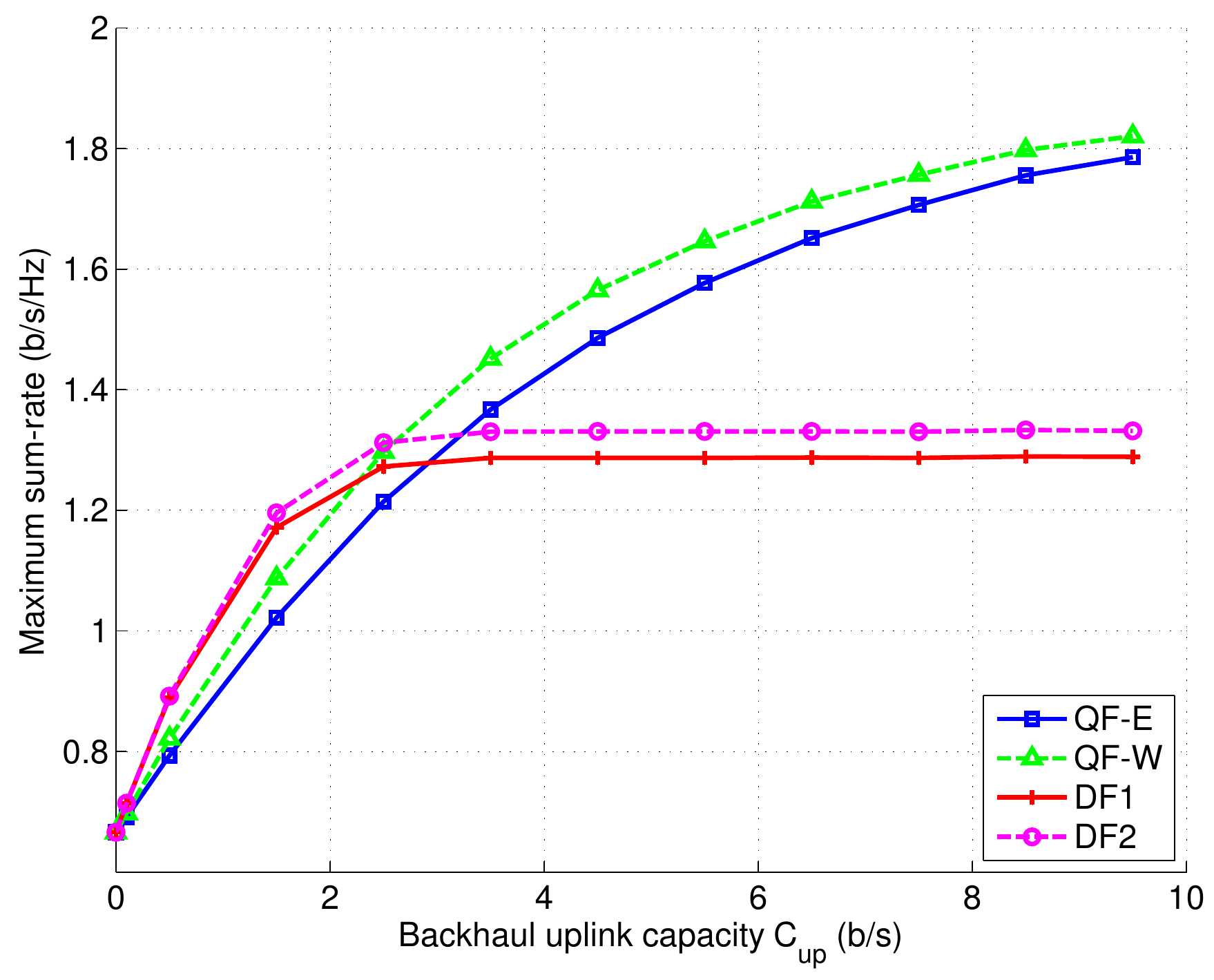}
\caption{Maximum sum--rate with $C_\mathrm{down} = C_\mathrm{up}$, $s_o = 150$m and varied $C_\mathrm{up}$.}
\label{backhaul_sumrate}
\end{figure}
\begin{figure}
\centering
\includegraphics[width=0.6\columnwidth]{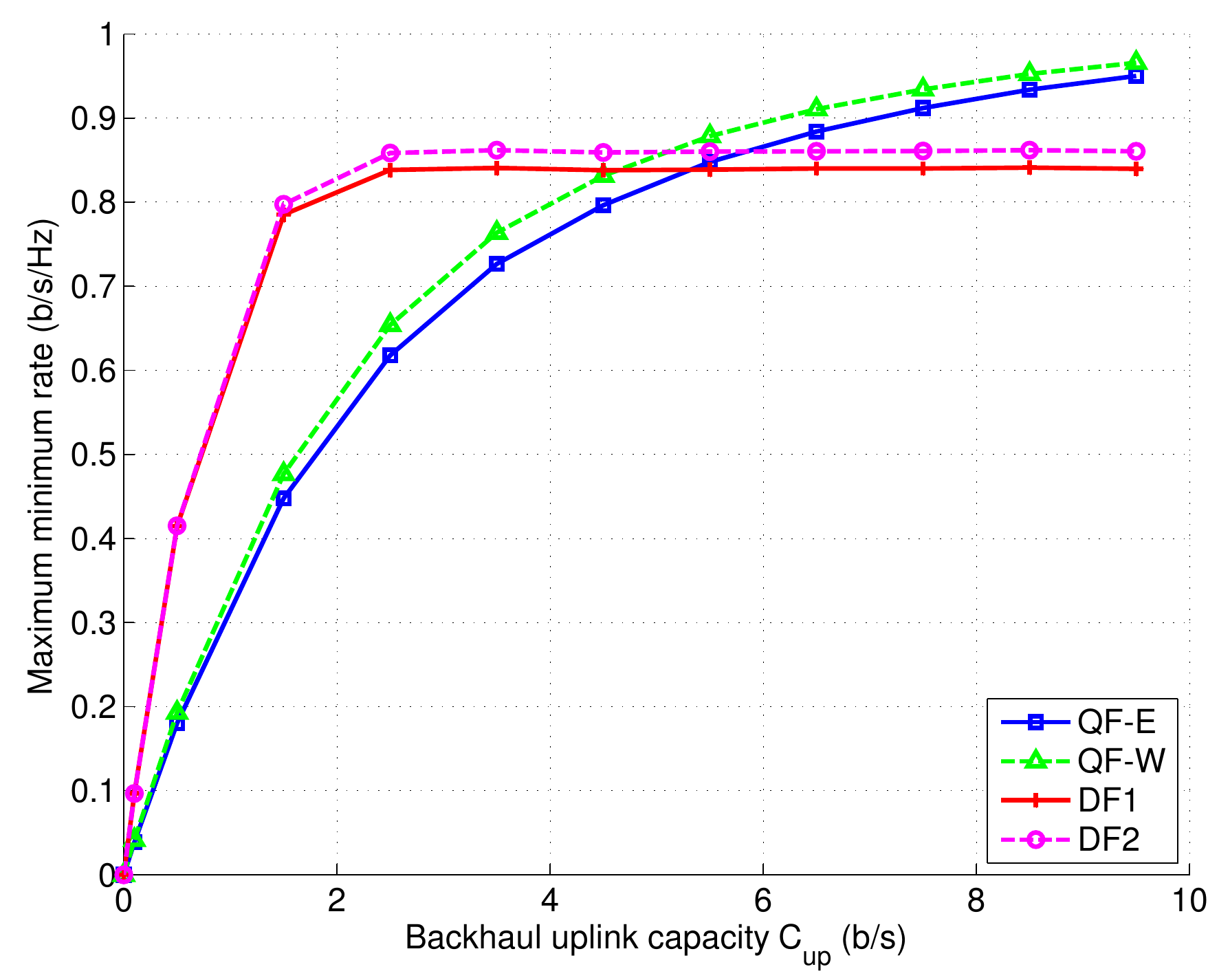}
\caption{Maximum minimum rate with $C_\mathrm{down} = C_\mathrm{up}$, $s_o = 150$m and varied $C_\mathrm{up}$.}
\label{backhaul_equalrate}
\end{figure}
\begin{figure}
\centering
\includegraphics[width=0.6\columnwidth]{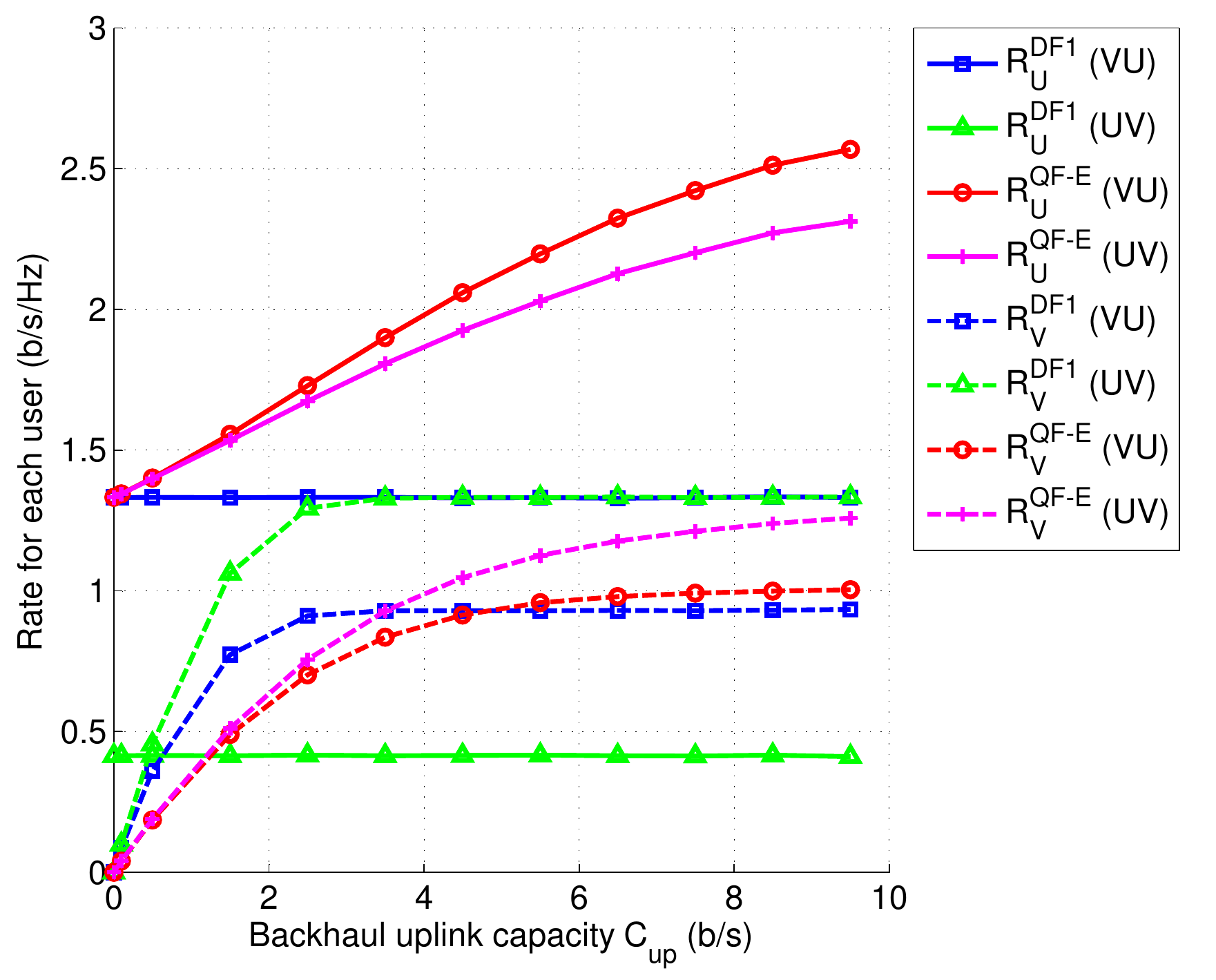}
\caption{Rate for each user with $C_\mathrm{down} = C_\mathrm{up}$, $s_o = 150$m and varied $C_\mathrm{up}$.}
\label{backhaul_userrate}
\end{figure}
We use Monte Carlo simulation to analyse the maximum sum-rate and the minimum rate of two users with cell radii $R_1 = 200$m, $R_2 = 20$m, path loss exponent $\alpha = 3$. In the first simulation, we vary $C_\mathrm{up}$ with fixed position of the FBS at $s_o = 150$m. The case with downlink backhaul, the DFQSI scheme, we assume $C_\mathrm{down} = 3C_\mathrm{up}$ to reflect the asymmetry of an ADSL link. Fig. \ref{backhaul_sumrate} and \ref{backhaul_equalrate} show the maximum sum-rate and the minium rate of the two users respectively. The DF schemes have a higher sum-rate when $C_\mathrm{up}$ is low and a lower sum-rate than the QF schemes when $C_\mathrm{up}$ is high. Let us compare the DF and QF schemes in the case of VU decoding order. As seen in (\ref{rv_df1_vu}), (\ref{ru_df1_vu}) and (\ref{ruv_qf_vu}), in most cases U is far away from the FBS thus $\gamma_{UF}$ is small and consequently $R^{QF}_U \approx R^{DF}_U$. It means that $R_U$ almost does not change and depends only on $\gamma_{UB}$. The difference between the DF and QF schemes therefore depends on the difference on $R_V$. At low $C_\mathrm{up}$, decoding at the FBS, the DF schemes avoid the quantization noise, which is high due to low $C_\mathrm{up}$ as seen in (\ref{backhaul_userrate}). The rate increases until reaching $\frac{\gamma_{VF}}{\gamma_{UF} + 1}$ and is upper bounded by this value. At high $C_\mathrm{up}$, the quantization noise, characterized by $\beta$, becomes negligible, the QF therefore continues to increase. The case with decoding order of UV can be explained in a similar way.

The superiority of a DF scheme compared to a QF scheme at low $C_\mathrm{up}$ in terms of the maximum minimum user rate is more obvious than in terms of sum--rate. This can be explained as follows. By quantizing the received signal and sending it to the MBS, in a QF scheme, the FBS actually facilitates the decoding of both $x_U$ and $x_V$ at the MBS for both decoding orders while, in a DF scheme, the backhaul uplink improves the decoding of $x_V$ only. Consequently, as in seen in Fig. \ref{backhaul_equalrate}, in both decoding orders of a QF scheme, $R_U$ is higher than $R_V$ while in a DF scheme, $R_U$ is higher than $R_V$ only in the decoding order of VU. This makes the maximum minimum user rate in a QF scheme lower at low $C_\mathrm{up}$.
\begin{figure}
\centering
\includegraphics[width=0.6\columnwidth]{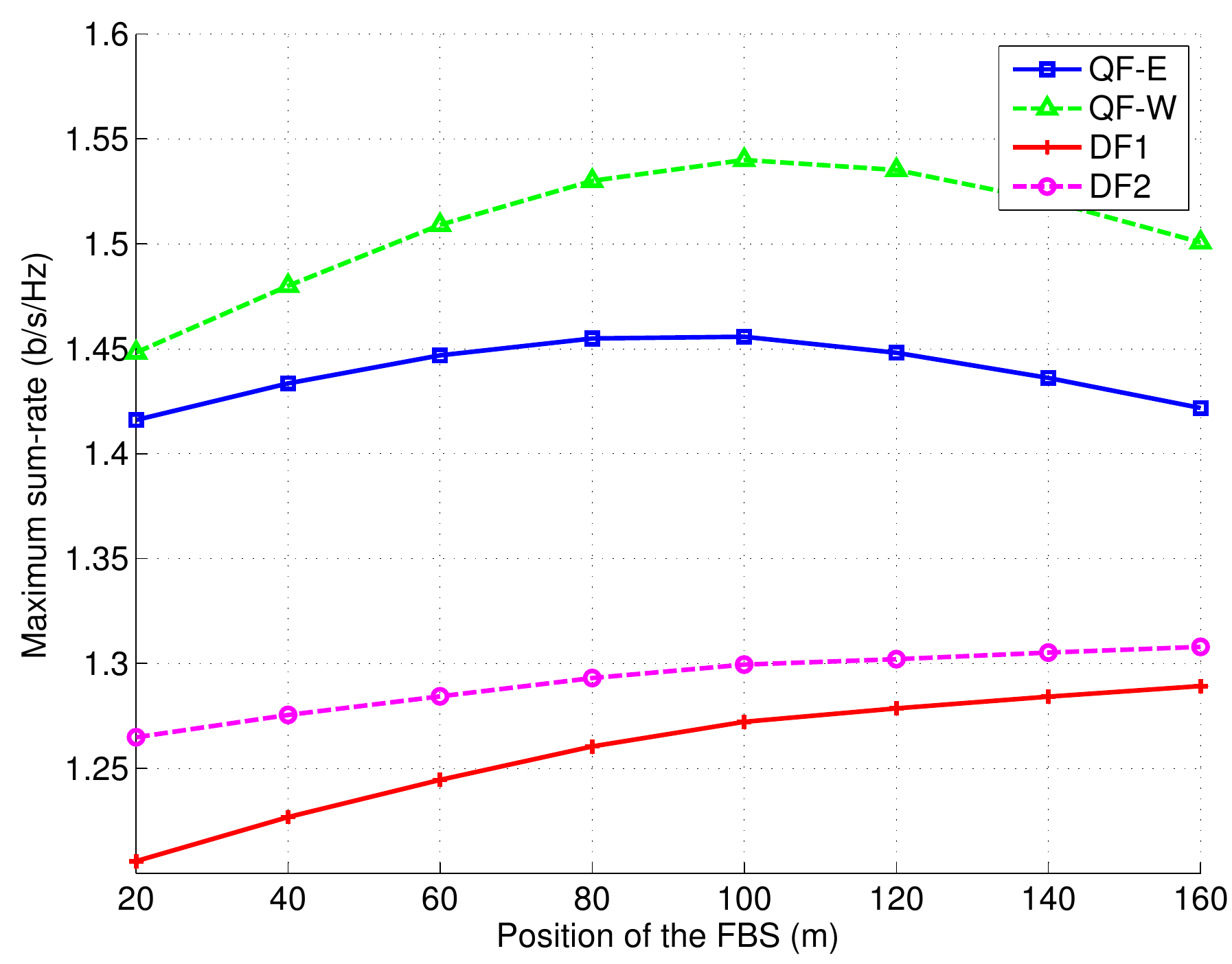}
\caption{Maximum sum--rate with $C_\mathrm{down} = 1$b/s/Hz, $C_\mathrm{up} = 4$b/s/Hz and varied $s_o = 150$.}
\label{fempos_sumrate}
\end{figure}
\begin{figure}
\centering
\includegraphics[width=0.6\columnwidth]{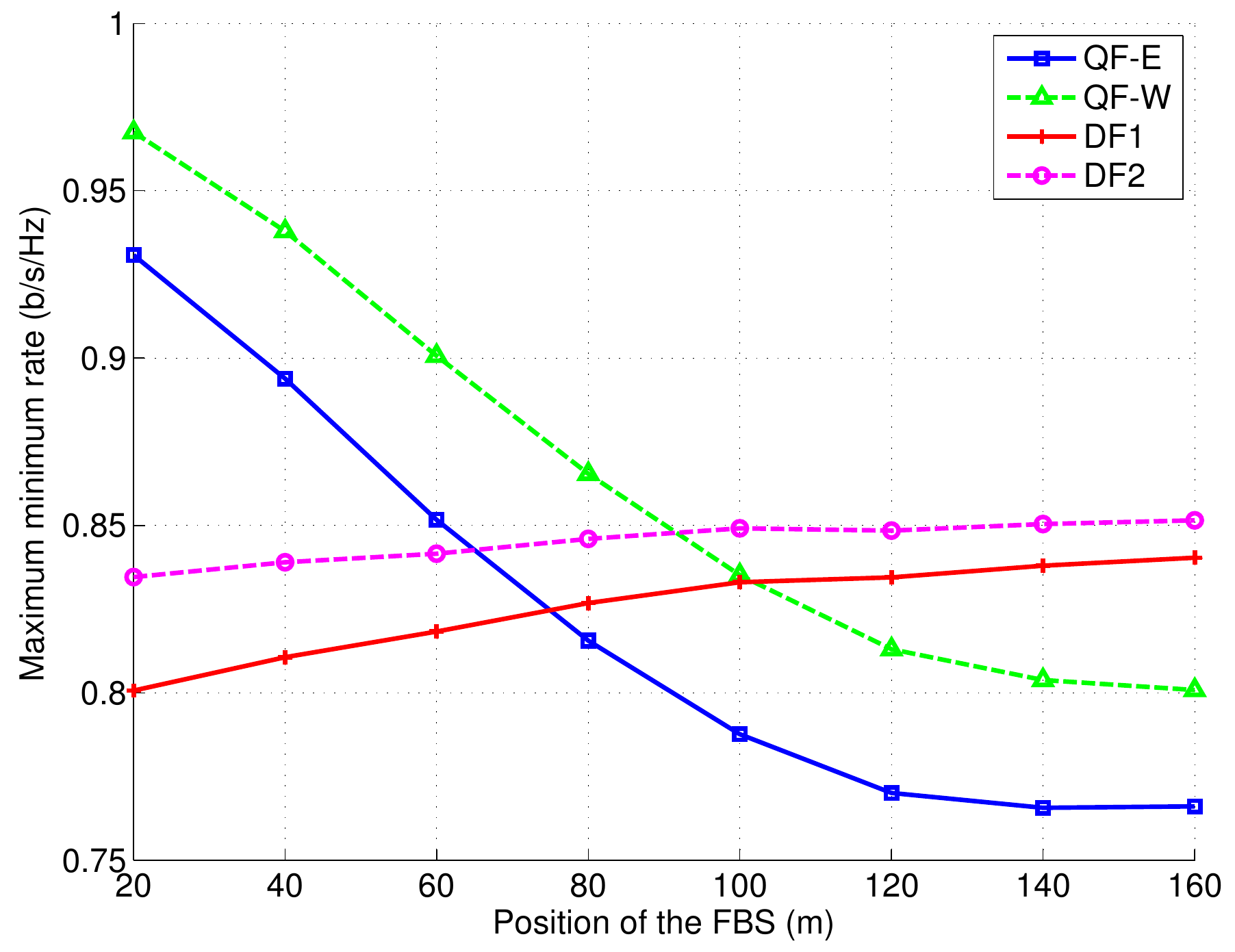}
\caption{Maximum minimum rate with $C_\mathrm{down} = 1$b/s/Hz, $C_\mathrm{up} = 4$b/s/Hz and varied $s_o = 150$.}
\label{fempos_equalrate}
\end{figure}
\begin{figure}
\centering
\includegraphics[width=0.6\columnwidth]{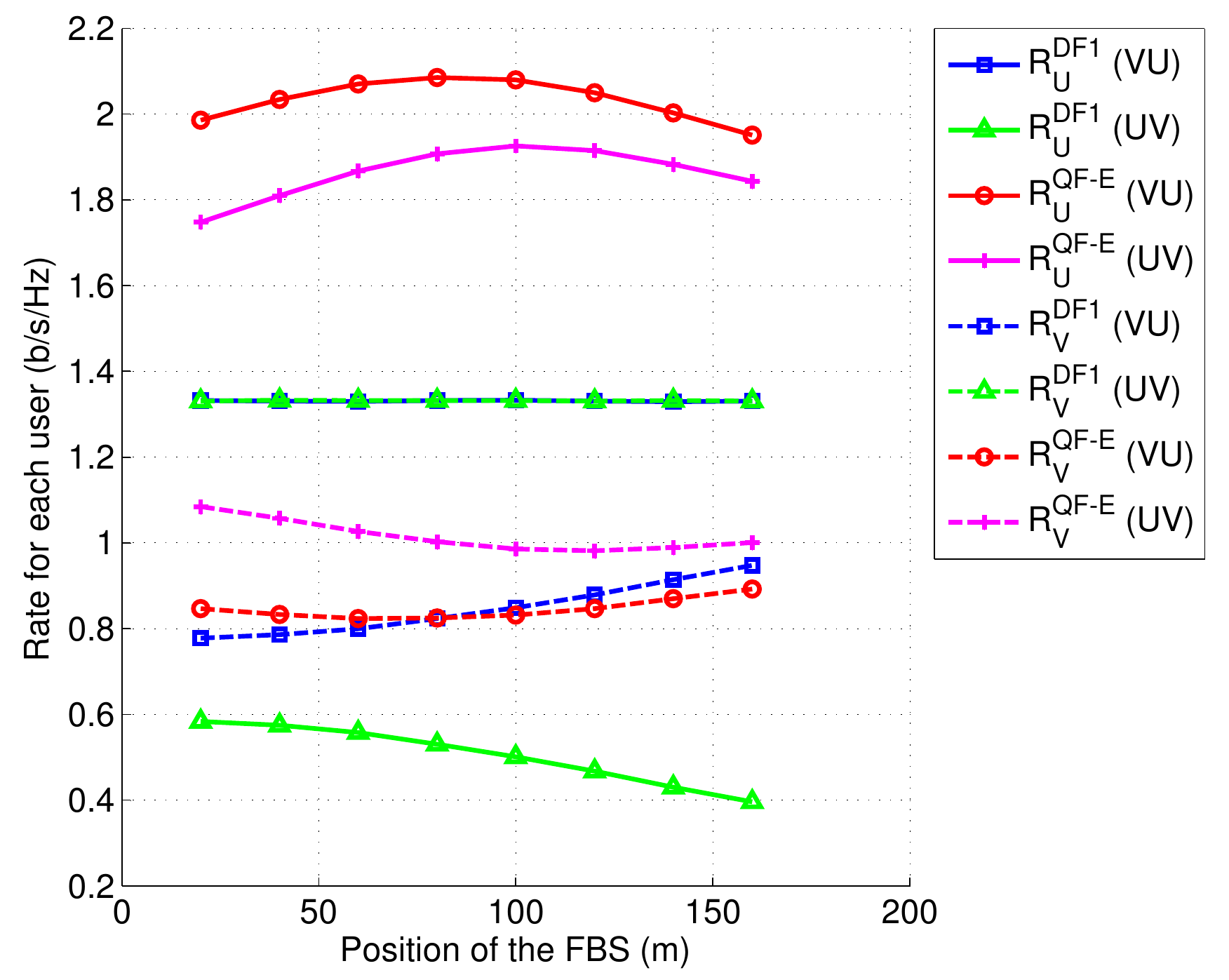}
\caption{Rate for each user with $C_\mathrm{down} = 1$b/s/Hz, $C_\mathrm{up} = 4$b/s/Hz and varied $s_o = 150$.}
\label{fempos_userrate}
\end{figure}

In the second simulation, we vary $s_o$ with fixed $C_\mathrm{up} = 4$ b/s/Hz. The case with downlink backhaul, the DFQSI scheme, we assume $C_\mathrm{down} = 1$ b/s/Hz to see the effect of a low-rate backhaul downlink. Fig. \ref{fempos_sumrate}, \ref{fempos_equalrate} and \ref{fempos_userrate} show the maximum sum-rate, the minimum rate of the two users and rate for each user respectively. When $s_o$ increases, the FBS goes to the edge of the cell, being uniformly randomized throughout the macrocell, on average, U becomes farther away from the FBS, while the U-MBS and V-FBS distances do not change thus $\gamma_{UF}$ decreases while $\gamma_{UB}$ and $\gamma_{VF}$ do not change. The change of $\gamma_{UF}$ affects the performance as follows. At high $s_o$, because $\gamma_{UF}$ is small, the decoding order of VU at the FBS is better than the decoding order of UV. Consequently, the sum-rate of the DF schemes increases correspondingly. On the other hand, in the QF schemes, $\gamma_{UF}$ plays a different role as it appears in both numerator and denominator of the SNR as seen in (\ref{ruv_qf_uv}) and (\ref{ruv_qf_vu}), the sum--rate therefore achieves a local maximum at a certain point. The minimum rate also has a similar tendency.
\section{Conclusion}
\label{conclusion}
We have proposed two novel quantization schemes, QF and DFQSI, which increase the performance of a two-tier network in terms of maximum sum--rate and minimum rate of two users. The comparison with a conventional scheme, DF, is considered and analyzed. The results show that the QF gives a much higher sum--rate compared to the DF schemes when the capacity of the backhaul uplink is high. If the backhaul uplink has a low capacity, the DF schemes are superior to the QF scheme. Therefore a better scheme can be chosen according to the capacity. On the other hand, the backhaul downlink also increases the sum--rate, in DFQSI, compared to the case it is not used at all, in DF.

\end{document}